\documentclass[a4paper]{article}
\usepackage{fullpage}

\usepackage[activate={true,nocompatibility},final,tracking=true,kerning=true,spacing=true,factor=1100,stretch=10,shrink=10]{microtype}
\microtypecontext{spacing=nonfrench}

\usepackage{lmodern}
\usepackage{amsthm}
\newtheorem{lemma}{Lemma}
\usepackage[english]{babel}
\usepackage{rotating}
\usepackage{amsmath}
\usepackage{url}
\usepackage[pdfborder={0 0 0}]{hyperref}
\usepackage{cite}
\usepackage{graphicx}
\usepackage{booktabs}
\usepackage{subfig}
\usepackage{adjustbox}

\usepackage{pgfplots}
\usetikzlibrary{pgfplots.groupplots}
\usetikzlibrary{matrix}
\pgfplotsset{compat=1.9}

\tikzset{adjust near node color/.code={%
    \pgfplotscolormapdefinemappedcolor\pgfplotspointmetatransformed%
    \definecolor{mapped node color}{rgb}{\pgfmathresult}%
    \pgfkeys{/tikz/text=mapped node color!70!black}%
    }
}

\usepackage[algo2e,ruled,vlined,linesnumbered,noend]{algorithm2e}
\DontPrintSemicolon
\SetAlCapHSkip{0pt}
\SetCommentSty{textit}
\SetFuncSty{}
\SetArgSty{}

\SetKwProg{Function}{function}{}{}
\SetKwFunction{BitRev}{Reverse}
\SetKwFunction{PrefixSum}{PrefixSum}
\SetKwFunction{FBit}{Bit}
\SetKwFunction{BitPrefix}{BitPrefix}
\SetKwFunction{rank}{rank}
\SetKwFunction{select}{select}
\SetKwFunction{access}{access}
\SetKwFunction{CountingSort}{ParallelCountingSort}

\newcommand{\UnaryOperator}[2][]{%
  \ifx&#1&%
  \ensuremath{\mathop{}\mathopen{}#2\mathopen{}}%
  \else%
  \ensuremath{\mathop{}\mathopen{}#2\mathopen{}\left(#1\right)}%
  \fi%
}

\newcommand{\UnaryArray}[2][]{%
  \ifx&#1&%
  \ensuremath{\mathop{}\mathopen{}#2\mathopen{}}%
  \else%
  \ensuremath{\mathop{}\mathopen{}#2\mathopen{}\lbrack#1\rbrack}%
  \fi%
}

\newcommand{\Oh}[1]{\UnaryOperator[#1]{\mathcal{O}}}
\newcommand{\WM}{\textsf{WM}}
\newcommand{\WT}{\textsf{WT}}

\newcommand{\NBRP}[1]{\ensuremath{\pi_{#1}}}
\newcommand{\Text}[1]{\UnaryArray[#1]{\mathsf{T}}}
\newcommand{\Hist}[1]{\UnaryArray[#1]{\mathsf{Hist}}}
\newcommand{\Borders}[1]{\UnaryArray[#1]{\mathsf{SPos}}}
\newcommand{\BV}[1]{\UnaryArray[#1]{\mathsf{BV}}^{\prime}}
\newcommand{\BVT}[1]{\ensuremath{\mathsf{BV}_{#1}}}
\newcommand{\Zeros}[1]{\UnaryArray[#1]{\mathsf{Z}}}

\newcommand{\BitPre}[1]{\UnaryOperator[#1]{\mathop{\text{prefix}}}}
\newcommand{\Bits}[1]{\UnaryOperator[#1]{\mathop{\text{bits}}}}
\newcommand{\Bit}[1]{\UnaryOperator[#1]{\mathop{\text{bit}}}}
\newcommand{\Reverse}[1]{\UnaryOperator[#1]{\mathop{\text{reverse}}}}

\newtheorem{observation}{Observation}
\newcommand*\samethanks[1][\value{footnote}]{\footnotemark[#1]}


\makeatletter
\newcommand{\removelatexerror}{\let\@latex@error\@gobble}
\makeatother

\begin{document}

\title{\Large Simple, Fast and Lightweight Parallel Wavelet Tree Construction\thanks{This work was supported by the German Research Foundation (DFG), priority programme ``Algorithms for Big Data'' (SPP 1736).}}

\author{Johannes Fischer\thanks{Technische Universit{\"a}t Dortmund, Department of Computer Science, \href{mailto:johannes.fischer@cs.tu-dortmund.de}{johannes.fischer@cs.tu-dortmund.de}, \href{mailto:florian.kurpicz@tu-dortmund.de}{florian.kurpicz@tu-dortmund.de}, \href{mailto:marvin.loebel@tu-dortmund.de}{marvin.loebel@tu-dortmund.de}}\\ \and Florian Kurpicz\samethanks\\ \and Marvin L{\"o}bel\samethanks}

\date{}
\maketitle

\begin{abstract}
  The wavelet tree (Grossi et al.~[SODA, 2003]) and wavelet matrix (Claude et al.~[Inf. Syst.,~47:15--32,~2015]) are compact indices for texts over an alphabet $[0,\sigma)$ that support \emph{rank}, \emph{select} and  \emph{access} queries in $O(\lg \sigma)$ time.
  We first present new practical sequential and parallel algorithms for wavelet tree construction.
  Their unifying characteristics is that they construct the wavelet tree \emph{bottom-up}, i.\,e., they compute the last level first.
  We also show that this bottom-up construction can easily be adapted to wavelet \emph{matrices}.
  In practice, our best sequential algorithm is up to twice as fast as the currently fastest sequential wavelet tree construction algorithm (Shun~[DCC,~2015]), simultaneously saving a factor of 2 in space.
  This scales up to 32 cores, where we are about equally fast as the currently fastest parallel wavelet tree construction algorithm (Labeit et al.~[DCC,~2016]), but still use only about 75\,\% of the space.
  An additional theoretical result shows how to adapt any wavelet \emph{tree} construction algorithm to the wavelet \emph{matrix} in the same (asymptotic) time, using only little extra space.
\end{abstract}

\section{Introduction}
The \emph{wavelet tree} (\WT{}), introduced in 2003 by Grossi et al.~\cite{Grossi2003}, is a space-efficient data structure that can answer \emph{access}, \emph{rank}, and \emph{select} queries for a text over an alphabet $[0,\sigma)$ in \Oh{\lg\sigma} time, requiring just ${n\lceil\lg\sigma\rceil(1+o(1))}$ bits of space.
\WT{}s are used as a basic data structure in many applications, e.\,g., text indexing \cite{Grossi2003}, compression~\cite{Makris2012WaveletTreesSurvey,Grossi2011WaveletTreesTheoryToPractice}, and in computational geometry as an alternative to fractional cascading \cite{Maekinen2006SubstringSearch}.
More information on the history of wavelet trees and many more of their applications can be found in the survey articles by Ferragina et al.~\cite{Ferragina2009MyriadWaveletTrees} and Navarro~\cite{Navarro2014waveletTree}.

\subsection{Our Contributions.}
In this paper, we focus on the construction of wavelet \emph{trees}, but the reader should note that with some trivial modifications all our sequential and parallel algorithms work as well for wavelet \emph{matrices} (and are actually also implemented for both variants).
The highlights of our new algorithms are the following:
\begin{itemize}
\item We present the fastest sequential \WT{}-construction algorithms (\emph{pcWT} and \emph{psWT}) that are up to twice as fast as \emph{serialWT}~\cite{Shun2015}, the previously fastest implementation for wavelet trees.
\item Simultaneously, our new algorithms use much less space than all previous ones: on realistically sized alphabets, \emph{pcWT} uses almost no space in addition to the input and output, while \emph{psWT} uses only one additional array of the same size as the text. Previous ones such as \emph{serialWT} or \emph{recWT}~\cite{Labeit2016parallelWaveletTree} use at least twice as much additional space.
\item We parallelize our new algorithms, obtaining the fastest parallel \WT{}-construction algorithms on medium-sized workstations of up to 32 cores.\footnote{Using more than 32 cores, \emph{recWT}~\cite{Labeit2016parallelWaveletTree} (the previously fastest parallel \WT{}-construction algorithm) remains faster.}
\item In particular, this results in the \emph{first} practical parallel algorithms for wavelet \emph{matrices}.
\end{itemize} 
A final (theoretical) contribution of this paper is that we show that the wavelet tree and the wavelet matrix are equivalent, in the sense that every algorithm that can compute the former can also compute the latter in the same time with only $(n+\sigma)(1+o(1))+(\sigma+2)\lceil\lg n\rceil$ bits of additional space.

\begin{figure}[t]
  \centering

  \begin{tabular}{cc}
    \adjustbox{valign=b}{\subfloat[Pointer-based wavelet tree.\label{sfig:example_pWT}]{%
      \includegraphics[scale=.66]{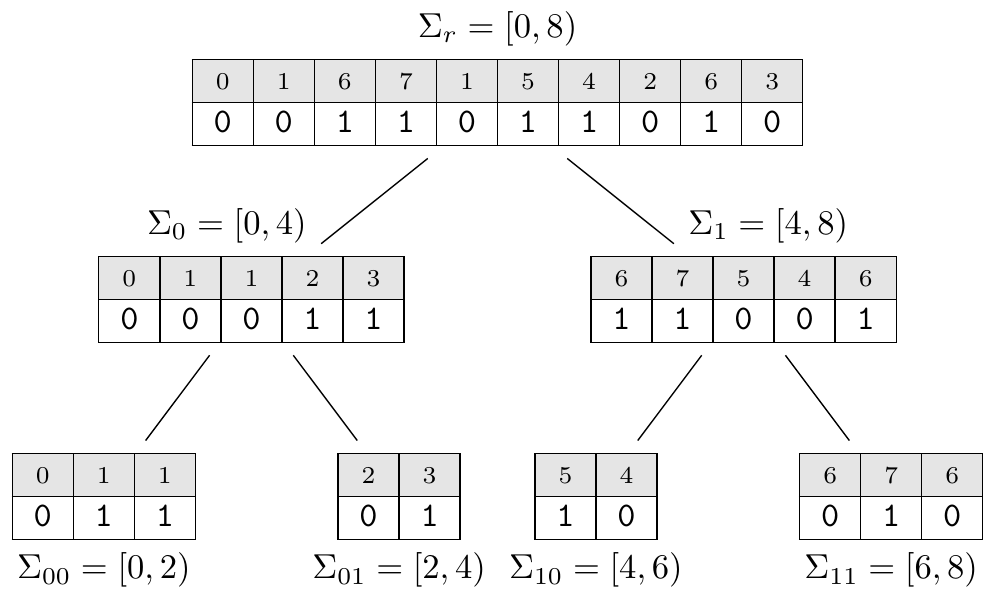}}}
    &      
    \adjustbox{valign=b}{\begin{tabular}{@{}c@{}}
    \subfloat[Level-wise wavelet tree.\label{sfig:example_lWT}]{%
      \includegraphics[scale=.66]{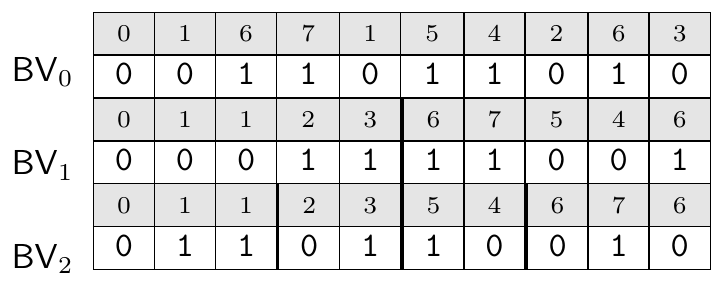}} \\
    \subfloat[Binary representation of \Text{}.\label{sfig:example_text}]{%
       \includegraphics[scale=.7]{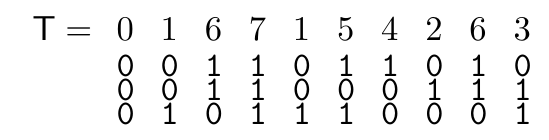}}
    \end{tabular}}
  \end{tabular}

  \caption{%
    The text $\Text{}=\texttt{0167154263}$, its binary representation in \protect\subref{sfig:example_text}, and the the two variants of wavelet trees of \Text{}.
    The light gray (\protect\tikz[baseline=.25ex]{ \fill[black!20,draw] (0, 0) rectangle (.3, .3); }) arrays contain the characters represented at the corresponding position in the bit vector and are not a part of the \WT{}.
    In \protect\subref{sfig:example_pWT}, $\Sigma_\alpha$ denotes the characters that are represented by the bit vector for $\alpha\in\lbrace r,0,1,00,01,10,11\rbrace$.
    In \protect\subref{sfig:example_lWT}, thick lines represent the borders of the intervals.
  \label{fig:example_wm_and_wt}}
\end{figure}

\subsection{Further Related Work.}
There exists lots of theoretical work when it comes to \WT{}-construction.
One line of research addresses lowering the construction time below \Oh{n\lg\sigma}, which is possible on a word-RAM by using word packing techniques.
Babenko et al.~\cite{Babenko2015WaveletTreeMeets} and Munro et al.~\cite{Munro2016FastWaveletTrees} independently obtained a construction time of \Oh{n\lceil\lg\sigma/\sqrt{\lg n}\rceil}.
Recently, Shun~\cite{Shun2016ImprovedParallelWaveletTree} has parallelized the word packing approach by Babenko et al.~\cite{Babenko2015WaveletTreeMeets} to improve the construction to \Oh{\sigma+\lg n} parallel time requiring \Oh{n\lceil\lg\sigma/\sqrt{\lg n}\rceil} work (here and in the following, we analyze parallel algorithms using  J{\'{a}}J{\'{a}}'s work-time paradigm \cite{JaJa1992Introduction}).

Fuentes-Sep{\'{u}}lveda et al.~\cite{Fuentes-Sepulveda2016parallelWaveletTree} were the first to describe and implement practical parallel \WT{}-construction algorithms, requiring \Oh{n} time and \Oh{n\lg\sigma} work.
Faster practical approaches were presented subsequently by Shun~\cite{Shun2015} and by Labeit et al.~\cite{Labeit2016parallelWaveletTree}, both requiring  \Oh{\lg n\lg\sigma} time and \Oh{n\lg\sigma} work.

A different line of research addresses the (theoretical) working space during construction:
Claude et al.~\cite{Claude2011WaveletTreeSpaceEfficient} and Tischler~\cite{Tischler2011WaveletTree} showed how to reduce the construction space for the \WT{} to $O(\lg n)$ bits.
However, none of these algorithms have been implemented beyond a proof-of-concept-status.

Although many papers (e.\,g., \cite{Shun2015,Shun2016ImprovedParallelWaveletTree}) on wavelet \emph{tree} construction mention that their algorithms can also be adapted to wavelet \emph{matrices}, none of them has actually been implemented.
The only (sequential and semi-external) implementation of a \WM{}-construction algorithm we are aware of is from the succinct data structure library (\emph{SDSL})~\cite{Gog2014}.
Finally, we mention that a faster and smaller \emph{alternative} to the \WT{} (that can only be used in very specific text indexing applications) can be constructed semi-externally~\cite{Gog2016FasterMinuter} and that there is a recent online \WT{}-construction algorithm~\cite{Fonseca2017WaveletTreeOnline}.

\section{Preliminaries}
Let $\mathsf{T}=\Text{0}\dots\Text{n-1}$ be a text of length $n$ over an alphabet $\Sigma=[0,\sigma)$.
Each character $\Text{i}$ can be represented using $\lceil\lg\sigma\rceil$ bits.
In this paper, the leftmost bit is the \emph{most significant bit} (MSB) and the \emph{least significant bit} (LSB) is the rightmost bit.
We denote the binary representation of a character $\alpha\in\Sigma$ as \Bits{\alpha}, e.\,g. $\Bits{3}=(\mathtt{011})_2$.
Whenever we write a binary representation of a value, we indicate it by a subscript two.
The $k$-th bit (from MSB to LSB) of a character $\alpha$ is denoted by $\Bit{k, \alpha}$ for all $0\leq k<\lceil\lg\sigma\rceil$.
Given $\alpha\in\Sigma$, the \emph{bit prefix} of size $k$ of $\alpha$ are the $k$ most significant bits, i.\,e., $\BitPre{k,\alpha}=(\Bit{0, \alpha}\dots\Bit{k-1,\alpha})_2$.
We interpret sequences of bits as integer values.

Let \BVT{} be a bit vector of size $n$.
The operation $\rank_0(\BVT{}, i)$ returns the number of 0's in $\BVT{}[0,i)$, whereas $\select_0(\BVT{}, i)$ returns the position of the $i$-th 0 in \BVT{}.
The operations $\rank_1(\BVT{}, i)$ and $\select_1(\BVT{}, i)$ are defined analogously.

Given an array $\textsf{A}$ of $n$ integers and an associative operator $+$ (we only use addition), the zero based \emph{prefix sum} for $\textsf{A}$ returns an array $\textsf{B}[0,n)$ with $\textsf{B}[0]=0$ and $\textsf{B}[i]=\textsf{A}[i-1]+\textsf{B}[i-1]$ for all $i\in[1,n)$.\footnote{If not zero based, $\textsf{B}$ is usually defined as $\textsf{B}\lbrack 0\rbrack=\textsf{A}\lbrack 0\rbrack$ and $\textsf{B}\lbrack i\rbrack=\textsf{A}\lbrack i-1\rbrack+\textsf{B}\lbrack i-1\rbrack$ for all $i\in[1,n)$.}
The prefix sum can be computed in \Oh{\lg n} parallel time and \Oh{n} work~\cite{JaJa1992Introduction}.

\subsection*{Wavelet Trees.}
Let \Text{} be a text of length $n$ over an alphabet $[0,\sigma)$.
The \emph{wavelet tree} (\WT{}) of \Text{} is a complete and balanced binary tree.
Each node of the \WT{} represents characters in $[\ell, r)\subseteq [0,\sigma)$.
The root of the \WT{} represents characters in $[0,\sigma)$, i.\,e., all characters.
The left (or right) child of a node representing characters in $[\ell, r)$ represents the characters in $[\ell, (\ell+r)/2)$ (or $[(\ell+r)/2,r)$, respectively).
A node is a leaf if $l+2 \geq r$.

The characters in $[\ell,r)$ at a node $v$ are represented using a bit vector $\BVT{v}$ such that the $i$-th bit in $\BVT{v}$ is $\Bit{d(v),{\Text{}_{[\ell,r)}}\lbrack i\rbrack}$, where $d(v)$ is the depth of $v$ in \WT{}, i.\,e., the number of edges on the path from the root to $v$, and $\Text{}_{[\ell,r)}$ denotes the array containing the characters of \Text{} (in the same order) that are in $[\ell,r)$.
The interval of a \WT{} at which a character is represented at level $\ell$ is encoded by its length-$\ell$ bit prefix, as shown in the following Observation:
\begin{observation}[Fuentes-Sep{\'u}lveda et al.~\cite{Fuentes-Sepulveda2014}]\label{obs:interval_wt}
  Given a character $\Text{i}$ for $i\in[0,n)$ and a level $\ell\in[1,\lceil\lg\sigma\rceil)$ of the \WT{}, the interval pertinent to $\Text{i}$ in $\BVT{\ell}$ can be computed by \BitPre{\ell,\Text{i}}.
\end{observation}

There are two variants of the \WT{}: the \emph{pointer-based} and the \emph{level-wise} \WT{}.
The pointer-based \WT{} uses pointers to represent the tree structure, see Figure~\ref{sfig:example_pWT}.
In the level-wise \WT{}, we concatenate the bit vectors of all nodes at the same depth in a pointer-based \WT{}.
Since we lose the tree topology, the resulting bit vectors correspond to a \emph{level} that is equal to the depth of the concatenated nodes.
We store only a single bit vector $\BVT{\ell}$ for each level $\ell\in[0,\lceil\lg\sigma\rceil)$, see Figure~\ref{sfig:example_lWT}.
This retains the functionality from the pointer-based \WT{}~\cite{Maekinen2006SubstringSearch,Maekinen2007RankAndSelect}, but reduces the redundancy for the binary rank- and select-structures on the bit vectors.

The wavelet tree (both variants) can be used to generalize the operations \access, \rank and \select from bit vectors to alphabets of size $\sigma$.
Answering these queries then requires \Oh{\lg\sigma} time.
To do so, the bit vectors are augmented by binary rank and select structures.
We point to \cite{Claude2015} for a detailed description of the operations.
In the following, we work with the level-wise \WT{}.

\section{New Wavelet Tree Construction Algorithms}
As shown in Observation~\ref{obs:interval_wt}, each level $\ell$ of the \WT{} contains disjoint intervals corresponding to the \mbox{length-$\ell$} bit prefixes of the characters in \Text{}.
This enables us to start on the \emph{last} level $\lceil\lg\sigma\rceil-1$, and then iteratively work through the other levels in a \emph{bottom-up} manner until the tree is fully constructed.
To get this process started, we need to know the borders of the intervals on the last level, for which we must first compute the \emph{histogram} of the text characters (as in the first phase of counting sort).
On subsequent levels $\ell\in[0,\lceil\lg\sigma\rceil-1)$ we use the fact that we can quickly compute the histograms of the considered bit prefixes of size $\ell$ from the histogram of bit prefixes of size $\ell+1$, \emph{without rescanning the text}.
Saving one scan of the text per level is one of the reasons that our algorithms are faster.
This and the resulting low memory consumption (up to 50\,\% of the competitors) are the main distinguishing features of our new algorithms from the previous \WT{}-construction algorithms.
We assume that arrays are initialized with 0's.
In this section, $id$ refers to the \emph{identity function}.
Later (when we construct wavelet matrices in \S{\ref{sec:the_wavelet_matrix}), we need to replace the identity function with the \emph{bit-reversal permutation}.

\subsection{Sequential Wavelet Tree Construction.}
\label{ssec:sequential_wavelet}
Our first \WT{}-construction algorithm (\emph{pcWT}, see Algorithm~\ref{alg:pcwm}) starts with the computation of the number of occurrences of each character in \Text{} to fill the initial histogram $\textsf{Hist}[0,\sigma)$.
In addition, the first level of the \WT{} is computed, as it contains the MSBs of all characters in text order (lines~\ref{alg:initial_histogram} and \ref{alg:compute_first_level}).
This requires \Oh{n} time and $\sigma\lceil\lg n\rceil$ bits space for the histogram.
Later on we require additional $\sigma\lceil\lg n\rceil$ bits to store the starting positions of the intervals (see array $\Borders{}[0,\sigma)$ in Algorithm~\ref{alg:pcwm}).

Initially, we have a histogram for all characters in \Text{}.
During each iteration (say at level $\ell$) we need the histogram for all bit prefixes of size $\ell - 1$ of the characters in \Text{}.
Therefore, if we have the histogram of length-$\ell$ bit prefixes, we can simply compute the histogram of the bit prefixes of size $\ell - 1$ by ignoring the last bit of the current prefix.
E.\,g., the amount of characters with bit prefix $(\mathtt{01})_2$ is the total number of characters with bit prefixes $(\mathtt{010})_2$ and $(\mathtt{011})_2$.
We can do so in \Oh{\sigma} time requiring no additional space reusing the space of the histogram of length-$\ell$ bit prefixes (line~\ref{alg:update_histogram}).

Using the updated histogram, we compute the starting positions of the intervals of the characters that can by identified by their bit prefix of size $\ell-1$ for level $\ell$.
The starting position of the interval representing characters with bit prefix $0$ is always $0$, therefore we only compute the starting positions for all other bit prefixes (line~\ref{alg:loop_compute_other_borders}).
Again, this requires \Oh{\sigma} time and no additional space, as we can reuse the space used to store the starting positions of the intervals of the previously considered level.

Last, we need to compute the bit vector for the current level $\ell$.
To do so, we simply scan \Text{} once from left to right and consider the bit prefix of length $\ell - 1$ of each character.
Since we have computed the the starting position (\Borders{}) in the bit vector where the $\ell$-th MSB of the characters needs to be stored, we can store it accordingly and increase the position for characters with the same bit prefix by one (lines~\ref{alg:get_position_by_bit_prefix_and_update}~and~\ref{alg:set_bit}).
This requires \Oh{n} time and no additional space.
Since we need to compute \Oh{\lg\sigma} levels, this results in the following Lemma:
\begin{lemma}
  Algorithm pcWT computes the \WT{} of a text of length $n$ over an alphabet of size $\sigma$ in \Oh{n\lg\sigma} time using $2\sigma\lceil\lg n\rceil$ bits of space in addition to the input and output.
\end{lemma}

\begin{figure}[t]
  \centering
  \begin{algorithm2e}[H]
    \small
    \caption{pcWT (sequential)\label{alg:pcwm}}
      \For{$i = 0$ \KwSty{to} $n - 1$}{
        $\Hist{\Text{i}}\textsf{++}$\label{alg:initial_histogram}\;
        $\BVT{0}\lbrack i\rbrack=\FBit(0,\Text{i})$\label{alg:compute_first_level}
      }
      \For{$\ell=\lceil\lg\sigma\rceil - 1$ \KwSty{to} $1$}{\label{alg:loop_other_levels_begin}
        \For{$i = 0$ \KwSty{to} $2^{\ell} - 1$}{
          $\Hist{i}=\Hist{2i}+\Hist{2i + 1}$\label{alg:update_histogram}
        }
        \For{$i = 1$ \KwSty{to} $2^{\ell} - 1$}{\label{alg:loop_compute_other_borders}
          \hangindent=1.25\skiptext\hangafter=1
          $\Borders{id(i)} = \Borders{id(i - 1)} + \Hist{id(i - 1)}$\label{alg:compute_other_borders}
        }
        \For{$i = 0$ \KwSty{to} $n - 1$}{
          $pos = \Borders{\BitPre{\ell,\Text{i}}}\textsf{++}$\label{alg:get_position_by_bit_prefix_and_update}\;
          $\BVT{\ell}\lbrack pos\rbrack=\Bit{\ell, \Text{i}}$\label{alg:set_bit}
        }
      }
  \end{algorithm2e}
\end{figure}

\subsection{Parallel Wavelet Tree Construction.}
\label{sec:parallel_wavelet_tree_construction}
The \emph{na\"{i}ve} way to parallelize the pcWT algorithm is to parallelize it such that each core is responsible for the construction of one level of the \WT{}.
To this end, each core needs to first compute the corresponding histogram of the level, and then the resulting starting positions of the intervals (each requiring $2^\ell\lceil\lg n\rceil$ bits of space at level $\ell$).
This results in the following Lemma:
\begin{lemma}
  The parallelization of pcWT computes the \WT{} in \Oh{n} time with \Oh{n\lg\sigma} work requiring $4\sigma\lceil\lg n\rceil$ bits of space in addition to the input and output.
\end{lemma}
The disadvantage of this na\"{i}ve parallelization is that we cannot efficiently use more than $\lceil\lg\sigma\rceil$ cores.
To use more cores, instead of parallelizing level-wise, we could do the following.
Each of the $p$ cores gets a slice of the text of size $\Theta(\frac{n}{p})$ and computes the corresponding bits in the bit vectors on \emph{all} levels.
On level $\ell$, each core $c$ first computes its \emph{local} histogram $\Hist{}_c[0,\sigma)$ according to the length-$\ell$ bit-prefixes of the input characters.
Using a parallel zero based prefix sum operation, these local histograms are then combined such that in the end each core knows where to write its bits (arrays $\Borders{}_c[0,\sigma)$ for $c\in [0,p)$).
As in the sequential algorithm, the final writing is then accomplished by scanning the local slice of the text from left to right, writing the bits to their correct places in $\BVT{\ell}$, and incrementing the corresponding value in $\Borders{}_c$.

This comes with the problem that two or more cores may want to concurrently write bits to the same computer word, resulting in \emph{race conditions}.
To avoid these race conditions, one would have to implement mechanisms for exclusive writes, which would result in unacceptably slow running times.
We rather propose the following approaches.

\subsubsection{Using Sorting.}
\label{ssec:text_decomposition}
Instead of having each core write randomly to each bit vector $\BVT{\ell}$, we want each core to be responsible for the same interval on each level of the \WT{}.
To this end, we \emph{globally} sort the input text (using the starting positions $\Borders{}_c$ on level $\ell$).
The resulting sorted text $\mathsf{T}_{\text{sorted}}$ is then again split into slices of size $\Theta(\frac{n}{p})$.
Then, each core scans its local slice from left to right and writes the corresponding bits to the bit vector $\BVT{\ell}$ (also from left to right).\footnote{Note that this is different from \emph{domain decomposition}, a popular approach for parallel \WT{}-construction~\cite{Labeit2016parallelWaveletTree,Fuentes-Sepulveda2016parallelWaveletTree} that we discuss in \S{\ref{sec:domain_decomposition}}.}
To avoid race conditions and \emph{false sharing}, i.\,e., working on data in a cache line that has been changed by another core, we further make sure that the size of each slice of the text is a common multiple of the cache lines' length and the size of a computer word.

The resulting parallel \WT{}-construction algorithm (\emph{psWT}, see Algorithm \ref{alg:pswm}) works as follows:
First, each of the $p$ cores computes the local histogram ($\Hist{}_c$ for $c\in[0,p)$) of its slice of \Text{} and, at the same time, fills $\BVT{0}$ (lines~\ref{alg:pswm_first_histogram} and \ref{alg:pswm_first_level}).
We compute the local starting positions ($\Borders{}_c$ for $c\in[0,p)$), using the zero based prefix sum of
$\Borders{}_0\lbrack 0\rbrack,\Borders{}_1\lbrack 0\rbrack,\dots,\Borders{}_{p-1}\lbrack 0\rbrack,\dots,\Borders{}_0\lbrack\sigma -1\rbrack,\dots,\Borders{}_{p-1}\lbrack\sigma - 1\rbrack$, with respect to (w.r.t.) $id$, see line~\ref{alg:pswm_second_prefix_sum}.
Here, ``w.r.t. $id$'' means that character $id(i)$ follows character $id(i - 1)$ for all $i\in[1,2^j)$.
Note that we replace $id$ with the bit-reversal permutation when constructing \WM{}s in \S{\ref{sec:adaption_of_our_algorithms}}.
All in all this requires \Oh{\lg p+\sigma} time, \Oh{n+p\sigma} work and $2p\sigma\lceil\lg n\rceil$ bits of space using $p$ cores.
Using this information ($\Hist{}_{c}$ and $\Borders{}_{c}$), we can compute the corresponding values of $\Hist{}_c$ and $\Borders{}_c$ for all levels $\ell\in[1,\lceil\lg\sigma\rceil)$.

For each level (see loop starting at line~\ref{alg:pswm_level_loop}) the time and work required are the same as during the first step.
There is no additional space required since we can reuse the space used during the previous iteration.
To sort the text, we use the local starting positions (to represent the intervals in counting sort, see line~\ref{alg:pswm_counting_sort}).
Storing the sorted text requires additional $n\lceil\lg\sigma\rceil$ bits of space (which we reuse at each level).
After sorting the text, each core can simply insert its bits at the corresponding position in $\BVT{\ell}$ (line \ref{alg:last_line_psWT}).
This leads to the following Lemma:
\begin{lemma}
  Algorithm psWT computes the \WT{} of a text of length $n$ over an alphabet of size $\sigma$ in \Oh{\lg\sigma\left(\frac{n}{p}+\lg p + \sigma\right)} time and \Oh{\lg\sigma(n+p\sigma)} work requiring $2p\sigma\lceil\lg n\rceil+n\lceil\lg\sigma\rceil$ bits of space in addition to the input and output using $p$ cores.
\end{lemma}
This algorithm can efficiently use up to $p\leq n/\sigma$ cores.
Using that many cores yields \Oh{n\lg\sigma} work with \Oh{\lg\sigma\left(\sigma +\lg n\right)} time.
Employing more cores would only increase the required work, without achieving a better running time than on $n/\sigma$ cores.
In theory, better work can be archived by using word packing techniques.
The algorithm can also be used to compute the \WT{} \emph{sequentially}, where it proved to be very efficiently (see~\S{\ref{sec:experiments}). 

Using sorting for the parallel construction of \WT{}s has already been considered by Shun~\cite{Shun2015} (\emph{sortWT}).
There, the \WT{} is computed from the first to the last level.
Hence, for each level the text has to be scanned twice for sorting and once (the sorted text) for the computation of the bit vector.

\begin{figure}[t]
  \centering
  \begin{algorithm2e}[H]
    \small
    \SetKwFor{ParFor}{parfor}{do}{endfor}
    \caption{psWT (parallel)\label{alg:pswm}}
      \ParFor{$c = 0$ \KwSty{to} $p - 1$}{
        \For{$i = c\frac{n}{p}$ \KwSty{to} $(c+1)\frac{n}{p}$}{
          $\Hist{}_c\lbrack\Text{i}\rbrack\textsf{++}$\label{alg:pswm_first_histogram}\;
          $\BVT{0}\lbrack i\rbrack=\Bit{0,\Text{i}}$\label{alg:pswm_first_level}
        }
      }
      \For{$\ell=\lceil\lg\sigma\rceil - 1$ \KwSty{to} $1$}{\label{alg:pswm_level_loop}
        \ParFor{$c = 0$ \KwSty{to} $p - 1$}{
          \For{$i = 0$ \KwSty{to} $2^{\ell} - 1$}{
            $\Hist{}_c\lbrack i\rbrack=\Hist{}_c\lbrack 2i\rbrack+\Hist{}_c\lbrack 2i+1\rbrack$
          }
        }
        $\Borders{}_c=$\,Parallel zero based prefix sum w.r.t. $id$\;\label{alg:pswm_second_prefix_sum}
        $\mathsf{T}_{\text{sorted}}=\CountingSort{\Text{},\,\Borders{}}$\label{alg:pswm_counting_sort}\;
        \ParFor{$c = 0$ \KwSty{to} $p - 1$}{
          \For{$i = c\frac{n}{p}$ \KwSty{to} $(c+1)\frac{n}{p}$}{\label{alg:start_the_same}
            $\BVT{\ell}\lbrack i\rbrack=\Bit{\ell, \mathsf{T}_{\text{sorted}}\lbrack i\rbrack}$\label{alg:last_line_psWT}
          }
        }
      }
  \end{algorithm2e}
\end{figure}

\subsubsection{Domain Decomposition.}
\label{sec:domain_decomposition}
The \emph{domain decomposition}~\cite{Labeit2016parallelWaveletTree,Fuentes-Sepulveda2016parallelWaveletTree} is a popular technique for parallel \WT{}-construction.
There, each core gets a slice of the text of size $\Theta(\frac{n}{p})$ and computes a \emph{partial} \WT{} for that slice (in parallel).
We use the sequential version of our \WT{}-construction algorithms \emph{pcWT} and \emph{psWT} (see \S{\ref{ssec:sequential_wavelet}} and \S{\ref{ssec:text_decomposition}}) to compute the partial \WT{}s (we call the resulting parallel algorithms \emph{ddpcWT} and \emph{ddpsWT}).
The final \WT{} is computed by merging all partial \WT{}s in parallel.

To merge the partial \WT{}s, we concatenate the intervals of all partial \WT{}s that correspond to the same bit prefix and store these concatenations with respect to their corresponding bit prefix at the correct level of the merged \WT{}.
We can do so in parallel by using the borders of the intervals of the partial \WT{}s that have already been computed during their construction.
To this end, a zero based prefix sum computes the starting positions of the intervals in the merged \WT{}.
Then, each processor writes its intervals at the corresponding positions.
Here, we also avoid race conditions by choosing the borders of the merged intervals according to the width of a computer word.
As the computation of the partial \WT{}s can be parallelized perfectly, we only require one parallel prefix sum, and the merging is one parallel scan of all bit vectors.
We do not merge in-place (and thus need another $n\lceil\lg\sigma\rceil$ bits for the final \WT{}).
When computing the partial \WT{}s with \emph{psWT}, we can reuse the space required for sorting the text.
This results in the following Lemma:

\begin{lemma}
  Algorithms ddpcWT and ddpsWT compute the \WT{} of a text \Text{} of length $n$ over an alphabet of size $\sigma$ in \Oh{\frac{n}{p}\lg\sigma+\lg p + \sigma} time and \Oh{n\lg\sigma+p\sigma} work requiring $2p\sigma\lceil\lg n\rceil+n\lceil\lg\sigma\rceil$ bits of space in addition to the input and output using $p$ cores.
\end{lemma}

\section{Experiments}
\label{sec:experiments}
We conducted our experiments on a workstation equipped with two Intel Xeon E5-2686 processor (22 cores with frequency up to 3\,GHz and cache sizes: 32\,kB L1D and L1I, 256\,kB L2 and 40\,MB L3) with Hyper-threading turned off and 256\,GB RAM.
We implemented our algorithms using C{}\verb!++!.
We compiled all code using \texttt{g}{}\verb!++! 6.2 with flags \texttt{-03} and \texttt{-march=native}.
To express parallelism, we use \emph{OpenMP}~4.5 in our algorithms.

\subsection{Algorithms.} In our experiments, we compare the implementations of the following algorithms (all sources have last been accessed on 2017-10-27):
\begin{itemize}
  \item \textbf{pcWT} and \textbf{psWT}: the new \WT{}-construction algorithms presented in this paper.
  We also parallelized these algorithm using \emph{domain decomposition} (\textbf{ddpcWT} and \textbf{ddpsWT}).\footnote{Available from \url{https://github.com/kurpicz/pwm}.}
  \item \textbf{serialWT}~\cite{Shun2016ImprovedParallelWaveletTree}: the previously fastest sequential \WT{}-construction algorithm that is based on~\cite{Fuentes-Sepulveda2014}.\footnote{Available from \url{https://people.csail.mit.edu/jshun}.}
  \item \textbf{levelWT}~\cite{Shun2016ImprovedParallelWaveletTree}: this algorithm constructs the \WT{} \emph{top-down} and determines the intervals similar to \emph{pcWT} but needs to scan the text twice for each level.\samethanks
  \item \textbf{recWT}~\cite{Labeit2016parallelWaveletTree}: the fastest parallel \WT{}-construction algorithm (when using more than 32 cores). Here, the text is split (in parallel) while computing the \WT{} top-down, such that each interval can be computed independently.\footnote{Available from \url{https://github.com/jlabeit/wavelet-suffix-fm-index}.}
  \item \textbf{ddWT} and \textbf{pWT}~\cite{Fuentes-Sepulveda2016parallelWaveletTree}: the original implementation of domain decomposition (\emph{ddWT}) and a parallel \WT{}-construction algorithm similar to \emph{levelWT}.\footnote{Available from \url{https://github.com/jfuentess/waveletree}.}
\end{itemize}
Summing up the state of the art prior to our work, \emph{serialWT} is the fastest sequential \WT{}-construction algorithm, and \emph{recWT} is the fastest parallel \WT{}-construction algorithm.
When it comes to memory usage, \emph{pWT} is the modest but up to 20 times slower than \emph{recWT}.
Due to the huge difference in running time, we have listed the results of our experiments for \emph{ddWT} and \emph{pWT} separately in Table~\ref{tab:fuentes_experiments}.
Other implementations (e.\,g. the \WM{}- and \WT{}-construction algorithms in the SDSL or \emph{sortWT}~\cite{Shun2015}) were already proved slower and/or more space consuming.

\subsection{Data Sets.}
For our experiments we use real-world texts and a text over a word-based alphabets, see Table~\ref{tab:description_of_texts} for more details.
All sources have last been accessed on 2017-10-27.
\begin{itemize}
  \item \textbf{XML}, \textbf{DNA}, \textbf{ENG}, \textbf{PROT} and \textbf{SRC}: texts from the \emph{Pizza and Chili} corpus containing XML documents, DNA data, English texts, protein data and source code.
  These files represent common real-world data ({\small\url{http://pizzachili.dcc.uchile.cl}}).
  \item \textbf{1000G}: collection of DNA data sets from the \emph{1000 Genomes Project}.
  This is an example of a text with a very small alphabet ({\small\url{http://www.internationalgenome.org/data}}).
  \item \textbf{CC}: concatenation of different websites (without the HTML tags) crawled by the \emph{common crawl} corpus. We removed all additional meta data, which has been added by the corpus ({\small\url{http://commoncrawl.org}}).
  \item \textbf{WORDS}: a collection of Russian news article from 2011 that we transformed in a word-based (integer) alphabet.
  This text is an example of a text with a large alphabet ({\small\url{http://statmt.org/wmt16/translation-task.html}}).
\end{itemize}

\begin{table}[t]
  \centering
  \footnotesize
  \begin{tabular}{cccccc}
    \toprule
    Name & $n/10^8$ & $\sigma$ & Name & $n/10^8$ & $\sigma$ \\
    \cmidrule[0.225ex](r){1-3}
    \cmidrule[0.225ex](l){4-6}
    XML            & $2.9$ & 97   & SRC        & $2.1$ & 230  \\
    \cmidrule(r){1-3}
    \cmidrule(l){4-6}
    DNA            & $4$ & 16     & 1000G      &  $88.2$ & 4 \\
    \cmidrule(r){1-3}
    \cmidrule(l){4-6}
    ENG            & $22.1$ & 239 & CC         & $100.7$ & 243 \\
    \cmidrule(r){1-3}
    \cmidrule(l){4-6}
    PROT           & $11.8$ & 27  & WORDS      & $1.4$& 2245405\\
    \bottomrule
  \end{tabular}
  \caption{%
    Statistics of the data used in our experiments.
    \label{tab:description_of_texts}}
\end{table}

\subsection{Results.}
Due to the structure of the paper we first focus on the \WT{}-construction algorithms, but the running times and the memory usage of our \WM{}-construction algorithms are nearly the same and can be found in \S{\ref{sec:the_wavelet_matrix}} (see Table~\ref{tab:wm_experiments}).
All running times are the median on five executions of the corresponding \WT{}-construction algorithm (without the construction of rank/select-support).
An overview of all running times and memory consumption can be found in Figure~\ref{fig:results_time_and_memory}.

\subsubsection{Running Times.}
In the sequential case, our new algorithm \emph{pcWT} and \emph{psWT} are of similar speed with \emph{psWT} being slightly faster than \emph{pcWT} being the second fastest.
On large alphabets \emph{pcWT} is 1.55 times as fast as \emph{psWT}, but on average \emph{psWT} is 2.75\,\% (and at most 9.62\,\%) faster than \emph{pcWT}.
Both algorithms are faster than the previously fastest \WT{}-construction algorithm \emph{serialWT}.
Compared with \emph{serialWT}, \emph{psWT} is on average 1.92 timer and at most 3.23 times as fast as \emph{serialWT}.
This results in a new fastest sequential \WT{}-construction algorithm that is also more memory efficient (see \S{\ref{sec:memory_consumption}}).

The situation is different in the parallel case (on 32 cores), where two algorithms (\emph{recWT}, \emph{ddpcWT}) are of similar speed.
On average \emph{ddpcWT} is 13\,\% faster than \emph{recWT}.
Especially on larger texts and texts with small alphabet (PROT and 1000G and CC), \emph{ddpcWT} is faster than \emph{recWT}.
On shorter texts and texts with large alphabets \emph{recWT} is faster than \emph{ddpcWT}, albeit \emph{pcWT} is of similar speed (but still slower).

Note that there is no distinct sequential version of our \emph{domain decomposition} algorithms as no merging is required and the \WT{} is constructed using \emph{pcWT} or \emph{psWT}.
On larger texts (e.\,g. 1000G and CC), our domain decomposition algorithms are faster than \emph{pcWT} and \emph{psWT}.
For really large alphabets, the domain decomposition algorithms are not well suited, as merging becomes very cost intensive for each level.

When it comes to small alphabets, the parallel version of \emph{pcWT} is not a good choice, as the number of cores that can be used is very small (we can only use 2 cores when computing the \WT{} for 1000G, see \S{\ref{sec:parallel_wavelet_tree_construction}}).
Furthermore, one of our presented algorithm (\emph{ddpcWT}) is of similar speed as the currently fastest parallel \WT{}-construction algorithm, while requiring less space.
Still, our algorithms do not scale as well as \emph{recWT}, see Figure~\ref{fig:speedup}.

The fast running times of our algorithms can be explained with the bottom-up construction.
Here, we require one scan less of the text per level than our competitors (except for \emph{recWT} that also requires only one scan of the text per level).

\begin{figure}
  \includegraphics{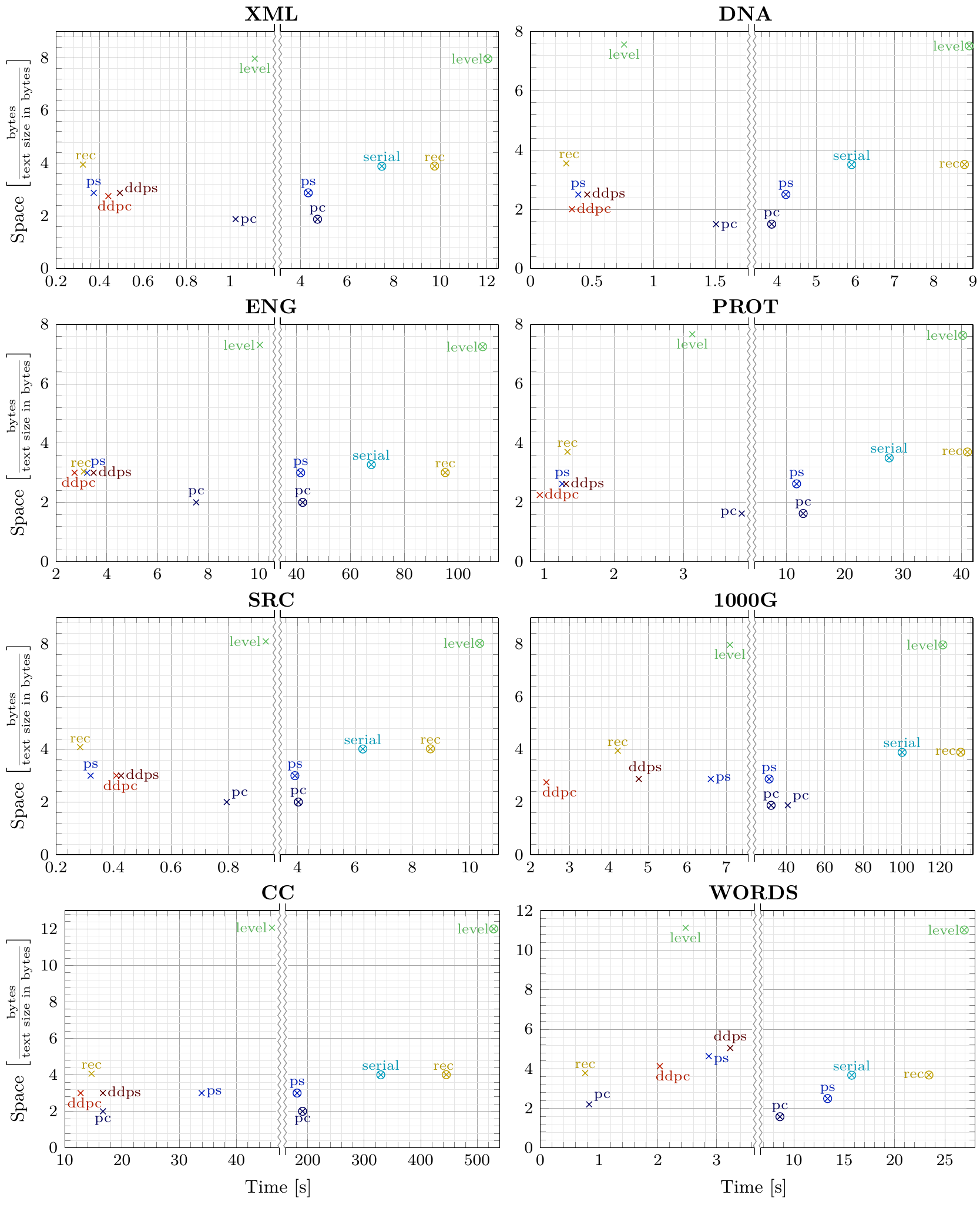}
  \caption{Running time and memory usage of the \WT{}-constuction algorithms measured in seconds and bytes per byte of the input text, resp.
  Algorithms run on one core are marked with $\otimes$ whereas algorithms running on 32 cores are marked with $\times$.\label{fig:results_time_and_memory}}
\end{figure}

\begin{figure}
  \includegraphics{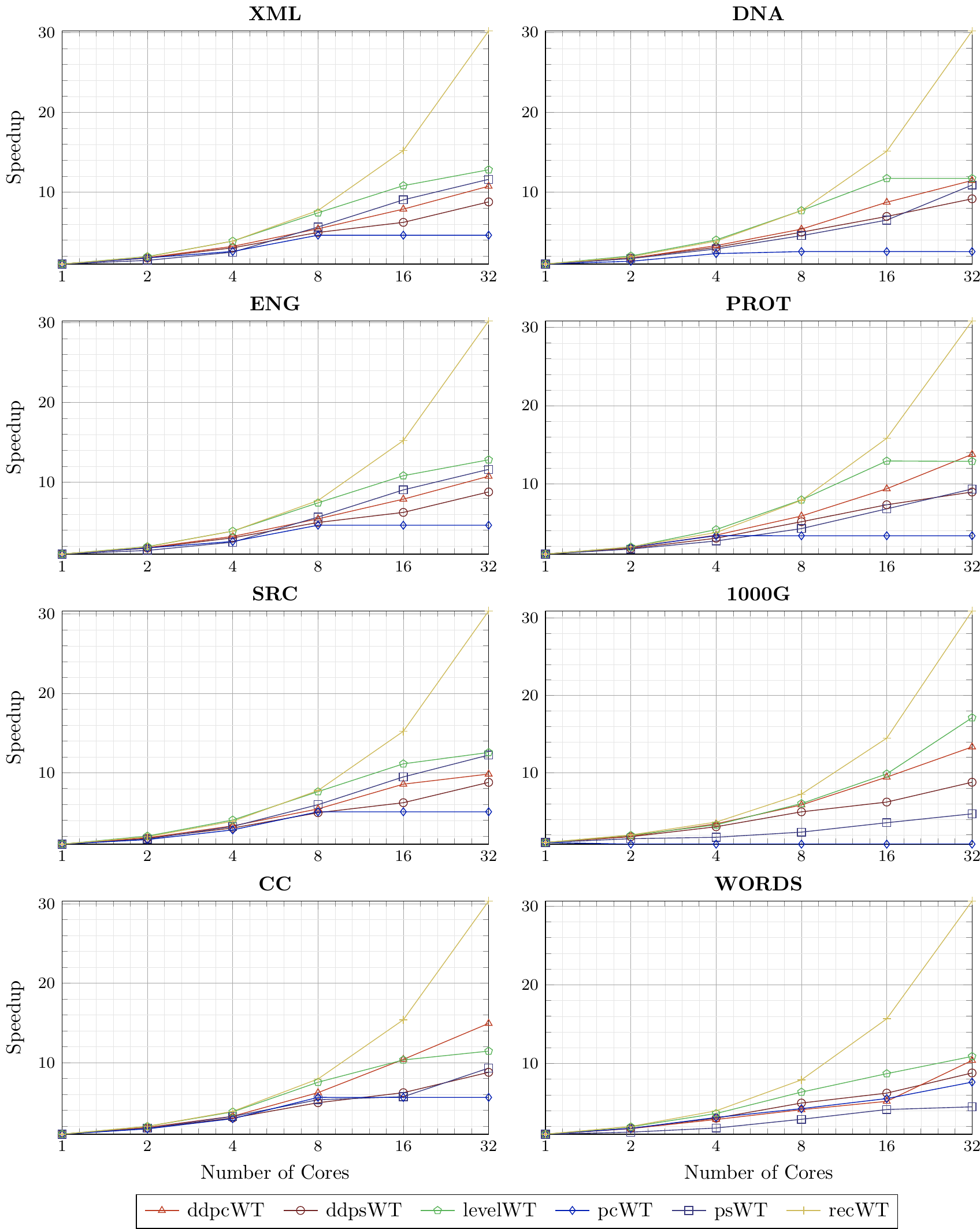}
  \caption{Comparison of the speedup of the \WT{}-construction algorithms.\label{fig:speedup}}
\end{figure}

\begin{table*}[t]
  \centering
  \scriptsize
  \begin{tabular}{ccccccccc}
    \toprule
      & \multicolumn{4}{c}{ddWT} & \multicolumn{4}{c}{pWT}\\
    \cmidrule[0.225ex](r){2-5}
    \cmidrule[0.225ex](l){6-9}
    Text & $t_1$ & $t_{32}$ & $\text{m}_1$ & $\text{m}_{32}$ & $t_1$ & $t_{32}$ & $\text{m}_1$ & $\text{m}_{32}$\\
    \cmidrule(r){1-1}\cmidrule(r){2-3}\cmidrule(lr){4-5}\cmidrule(lr){6-7}\cmidrule(l){8-9}
     XML    &14.231&5.078&2.815&2.783&13.574&2.450&1.944&1.966\\
    \cmidrule(r){1-1}\cmidrule(r){2-3}\cmidrule(lr){4-5}\cmidrule(lr){6-7}\cmidrule(l){8-9}
     DNA    &-&-&-&-&13.060&4.152&1.489&1.511\\
    \cmidrule(r){1-1}\cmidrule(r){2-3}\cmidrule(lr){4-5}\cmidrule(lr){6-7}\cmidrule(l){8-9}
     ENG    &136.866&8.909&2.987&2.966&132.871&21.094&1.993&1.994\\
    \cmidrule(r){1-1}\cmidrule(r){2-3}\cmidrule(lr){4-5}\cmidrule(lr){6-7}\cmidrule(l){8-9}
     PROT   &-&6.136&-&2.258&49.105&-&1.633&-\\
    \cmidrule(r){1-1}\cmidrule(r){2-3}\cmidrule(lr){4-5}\cmidrule(lr){6-7}\cmidrule(l){8-9}
     SRC    &12.620&5.176&3.073&3.025&12.056&1.869&2.091&2.124\\
    \cmidrule(r){1-1}\cmidrule(r){2-3}\cmidrule(lr){4-5}\cmidrule(lr){6-7}\cmidrule(l){8-9}
     1000G  &163.438&5.849&1.438&1.385&159.540&83.642&1.124&1.125\\
    \cmidrule(r){1-1}\cmidrule(r){2-3}\cmidrule(lr){4-5}\cmidrule(lr){6-7}\cmidrule(l){8-9}
     CC     &-&24.0159&-&2.994&624.458&90.914&1.401&1.402\\
    \cmidrule(r){1-1}\cmidrule(r){2-3}\cmidrule(lr){4-5}\cmidrule(lr){6-7}\cmidrule(l){8-9}
     WORDS  &26.386&9.563&2.786&2.803&26.666&3.673&1.869&1.883\\
    \bottomrule
  \end{tabular}
  \caption{%
    Experimental results of the \WT{}-construction algorithms \emph{ddWT} and \emph{pWT} \cite{Fuentes-Sepulveda2016parallelWaveletTree}.
    The experiments were conducted on the hardware and test instances that are described in~\S{\ref{sec:experiments}}.
    We measured the running time (in seconds) of the algorithms using one core ($t_1$) and 32 cores ($t_{32}$).
    The memory is given in bytes per byte of the input text when using one core ($\text{m}_1$) and 32 cores ($\text{m}_{32}$).
    A dash denotes that the algorithm could not compute the \WT{} of the given text.\label{tab:fuentes_experiments}}
\end{table*}

\subsubsection{Memory Consumption.}
\label{sec:memory_consumption}
The disadvantages of our algorithms when it comes to scaling are redeemed by their memory consumption, see again Figure~\ref{fig:results_time_and_memory}.
There we marked the number of bytes required per byte of input.
The lowest memory consumption is achieved by \emph{pcWT}, which matches our theoretical assumptions.
Next, \emph{psWT} requires 35\,\% more memory than \emph{pcWT}, but still 27\,\% less than \emph{recWT} when both are executed in parallel.
In the sequential case, \emph{pcWT} and \emph{psWT} require 50\,\% and 25\,\% less space than \emph{serialWT}.
Our \emph{domain decomposition} algorithms also match their expected memory consumption, as they require the same space as the algorithm used for the construction of the partial \WT{}s in addition to a bit vector of the size of the text used for merging the partial \WT{}s.
(If \emph{psWT} is used to compute the partial \WT{}s, the space used for sorting of the text slices can be reused for the merging.)
The memory consumption of \emph{levelWT} is enormous, requiring around 77\,\% more memory than \emph{pcWT} in both cases (sequential and parallel).

In practice, our algorithms require less memory than their competitors (with WORDS being the only exception).\footnote{The implementations by Fuentes-Sep{\'{u}}lveda et al.~\cite{Fuentes-Sepulveda2016parallelWaveletTree} require a similar amount of memory but are significantly slower.}
One reason is that our competitors use multiple arrays of text size to speed up the computation.

\section{The Wavelet Matrix}
\label{sec:the_wavelet_matrix}
A variant of the \WT, the \emph{wavelet matrix} (\WM{}), was introduced in 2011 by Claude et al.~\cite{Claude2015}.
It requires the same space as a \WT{} and has the same asymptotic running times for access, rank, and select; 
but in practice it is often faster than a \WT{} for rank and select queries~\cite{Claude2015}, as it needs less calls to binary rank/select data structures.
However, the fact that the \WM{} loses some nice structural properties of the \WT{} makes it harder to parallelize its construction, as \emph{divide-and-conquer} \WT{}-construction algorithms, e.\,g.~\emph{recWT}~\cite{Labeit2016parallelWaveletTree}, cannot simply be transformed to \WM{}s.

For the definition of the \WM{}, we need additional notations:
\emph{Reversing} the significance of the bits is denoted by \Reverse{}, e.\,g., $\Reverse{(\mathtt{001})_2}=(\mathtt{100})_2$.
The \emph{bit-reversal} permutation\footnote{\url{http://oeis.org/A030109}, last accessed 2017-10-27.} of order $k$ (denoted by \NBRP{k}) is a permutation of $[0,2^{k})$ with $\NBRP{k}(i)=(\Reverse{\Bits{i}})_2$.
For example, $\NBRP{2}=(0,2,1,3)=((\mathtt{00})_2,(\mathtt{10})_2,(\mathtt{01})_2,(\mathtt{11})_2)$.
\NBRP{k} and \NBRP{k+1} can be computed from another, as $\NBRP{k+1}=(2\NBRP{k}(0),\dots,2\NBRP{k}(2^{k}-1),2\NBRP{k}(0)+1,\dots,2\NBRP{k}(2^{k}-1)+1)$ and $\NBRP{k}=(\NBRP{k+1}(0)/2,\dots,\NBRP{k+1}(2^{k}-1)/2)$.
In practice, we can realize the division by a single bit shift.

\subsection*{Wavelet Matrices.}
The \emph{wavelet matrix} (\WM{})~\cite{Claude2015} has only a single bit vector $\BV{}_\ell$ per level $\ell\in[0,\lceil\lg\sigma\rceil)$ like the level-wise \WT{}, but the tree structure is discarded completely in the sense that we do not require each character to be represented in an interval that is covered by the character's interval on the previous level.
In addition, we use the array $\Zeros{}[0,\lceil\lg\sigma\rceil)$ to store the number of zeros at each level $\ell$ in \Zeros{\ell}.

$\BV{}_0$ contains the MSBs of each character in \Text{} in text order (this is the same as the first level of a \WT{}). For $\ell\ge 1$, $\BV{}_{\ell}$ is defined as follows.
Assume that a character $\alpha$ is represented at position $i$ in $\BV{}_{\ell-1}$. Then the position of its $\ell$-th MSB in $\BV{}_{\ell}$ depends on $\BV{}_{\ell-1}[i]$ in the following way:
if $\BV{}_{\ell-1}[i]=0$, \Bit{\ell,\alpha} is stored at position $\rank_{0}(\BV{}_{\ell-1},i)$; otherwise ($\BV{}_{\ell-1}[i]=1$), it is stored at position $\Zeros{\ell-1}+\rank_{1}(\BV{}_{\ell-1},i)$.
For an example, see Figure~\ref{fig:example_WM}.

Similar to the intervals in $\BVT{\ell}$ of the \WT{}, characters of \Text{} form intervals in $\BV{}_\ell$ of the \WM{}.
Again, the intervals at level $\ell$ correspond to bit prefixes of size $\ell$, but due to the construction of the \WM{} we consider the reversed bit prefixes.
The simplicity of the change required to turn the previously discussed \WT{}-construction algorithms in \WM{}-construction algorithms are based on the following Observation:
\begin{observation}\label{obs:intervals_wm}
  Given a character $\Text{i}$ for $i\in[0,n)$ and a level $\ell\in[1,\lceil\lg\sigma\rceil)$ of the \WM{}, the interval pertinent to $\Text{i}$ in $\BV{}_\ell$ can be computed by \Reverse{\BitPre{\ell,\Text{i}}}.
  Namely, $\BV{}_\ell\lbrack i\rbrack=\Bit{\ell, \Text{}^{\prime}\lbrack i\rbrack}$, i.\,e., the $\ell$-th MSB of the $i$-th character of $\Text{}^{\prime}$ in text order, where $\Text{}^{\prime}$ is $\Text{}$ stably sorted using the reversed bit prefixes of length $\ell$ of the characters as key.
\end{observation}

As with \WT{}s, if the bit vectors are augmented by (binary) rank and select data structures, the \WM{} can be used to answer \access, \rank and \select queries on a text over an alphabet of size $\sigma$ in \Oh{\lg\sigma} time.
We refer to \cite{Claude2015} for a detailed description of these queries.

\subsection{Adaption of our Algorithms to Wavelet Matrices.}
\label{sec:adaption_of_our_algorithms}
When comparing the bit vectors of the \WT{} and the \WM{} at level $\ell$, we see two similarities.
First, both bit vectors contain the $\ell$-th MSB of each character of \Text{} and second, the bits are grouped in intervals with respect to the bit prefix of size $\ell$ of the corresponding character and appear in the same order.
Thus, the number and sizes of the intervals is the same.
The difference is only the \emph{position} of the intervals within each level.
At level $\ell$, the intervals in $\BVT{\ell}$ of a \WT{} occur in increasing order with respect to the bit prefixes of size $\ell$ of the characters in \Text{}, i.\,e., the first interval corresponds to characters with bit prefix $0$, the second corresponds to characters with bit prefix $1$, and so on.
The intervals in $\BV{}_\ell$ of a \WM{} occur in increasing order with respect to the bit-reversal permutation \NBRP{\ell} of the characters in \Text{}.

All our algorithms (\emph{pcWT}, \emph{psWT}, \emph{ddpcWT} and \emph{ddpsWT}) can be adjusted to compute the \WM{} instead of the \WT{}.
We call them \emph{pcWM}, \emph{psWM}, \emph{ddpcWM} and \emph{ddpsWM}, respectively.
To do so, we just have to replace the identity permutation by the bit reversal permutation $\NBRP{}$, i.\,e., choosing $id=\NBRP{\ell}$ in lines \ref{alg:compute_other_borders} and \ref{alg:pswm_second_prefix_sum} in Algorithms \ref{alg:pcwm} and \ref{alg:pswm}, resp.
Then, the resulting starting positions of the intervals for bit prefixes are in bit reversal permutation order, i.\,e., the starting positions of the intervals for a \WM{} (compare Observations~\ref{obs:interval_wt} and~\ref{obs:intervals_wm}).
\begin{table}[t]
  \centering
  \scriptsize
  \begin{tabular}{ccccccccccccc}
    \toprule
      & \multicolumn{4}{c}{pc} & \multicolumn{4}{c}{ps} & \multicolumn{2}{c}{ddpc} & \multicolumn{2}{c}{ddpc}\\
    \cmidrule[0.225ex](r){2-5}
    \cmidrule[0.225ex](lr){6-9}
    \cmidrule[0.225ex](lr){10-11}
    \cmidrule[0.225ex](l){12-13}
    Text & $t_1$ & $t_{32}$ & $\text{m}_1$ & $\text{m}_{32}$ & $t_1$ & $t_{32}$ & $\text{m}_1$ & $\text{m}_{32}$ & $t_{32}$ & $\text{m}_{32}$ & $t_{32}$ & $\text{m}_{32}$\\
    \cmidrule(r){1-1}\cmidrule(r){2-3}\cmidrule(lr){4-5}\cmidrule(lr){6-7}\cmidrule(lr){8-9}\cmidrule(r){10-10}\cmidrule(lr){11-11}\cmidrule(lr){12-12}\cmidrule(l){13-13}
     XML    &4.737&0.988&1.875&1.875&4.355&0.531&2.875&2.875&0.468&2.750&0.529&2.875\\
    \cmidrule(r){1-1}\cmidrule(r){2-3}\cmidrule(lr){4-5}\cmidrule(lr){6-7}\cmidrule(lr){8-9}\cmidrule(r){10-10}\cmidrule(lr){11-11}\cmidrule(lr){12-12}\cmidrule(l){13-13}
     DNA    &3.895&1.479&1.500&1.500&4.293&0.904&2.500&2.500&0.691&2.000&0.805&2.500\\
     \cmidrule(r){1-1}\cmidrule(r){2-3}\cmidrule(lr){4-5}\cmidrule(lr){6-7}\cmidrule(lr){8-9}\cmidrule(r){10-10}\cmidrule(lr){11-11}\cmidrule(lr){12-12}\cmidrule(l){13-13}
     ENG    &42.705&7.388&2.000&2.000&41.847&4.203&3.000&3.000&3.252&3.000&3.975&3.000\\
     \cmidrule(r){1-1}\cmidrule(r){2-3}\cmidrule(lr){4-5}\cmidrule(lr){6-7}\cmidrule(lr){8-9}\cmidrule(r){10-10}\cmidrule(lr){11-11}\cmidrule(lr){12-12}\cmidrule(l){13-13}
     PROT   &12.820&3.780&1.625&1.625&11.695&1.438&2.625&2.625&1.112&2.250&1.344&2.625\\
     \cmidrule(r){1-1}\cmidrule(r){2-3}\cmidrule(lr){4-5}\cmidrule(lr){6-7}\cmidrule(lr){8-9}\cmidrule(r){10-10}\cmidrule(lr){11-11}\cmidrule(lr){12-12}\cmidrule(l){13-13}
     SRC    &3.968&0.750&2.000&2.000&3.796&0.488&3.000&3.000&0.371&3.001&0.414&3.001\\
     \cmidrule(r){1-1}\cmidrule(r){2-3}\cmidrule(lr){4-5}\cmidrule(lr){6-7}\cmidrule(lr){8-9}\cmidrule(r){10-10}\cmidrule(lr){11-11}\cmidrule(lr){12-12}\cmidrule(l){13-13}
     1000G  &32.322&4.900&1.250&1.250&30.908&6.969&2.250&2.250&2.429&2.250&4.750&2.250\\
     \cmidrule(r){1-1}\cmidrule(r){2-3}\cmidrule(lr){4-5}\cmidrule(lr){6-7}\cmidrule(lr){8-9}\cmidrule(r){10-10}\cmidrule(lr){11-11}\cmidrule(lr){12-12}\cmidrule(l){13-13}
     CC     &191.166&33.893&2.000&2.000&203.923&15.817&3.000&3.000&13.798&3.000&16.180&3.000\\
     \cmidrule(r){1-1}\cmidrule(r){2-3}\cmidrule(lr){4-5}\cmidrule(lr){6-7}\cmidrule(lr){8-9}\cmidrule(r){10-10}\cmidrule(lr){11-11}\cmidrule(lr){12-12}\cmidrule(l){13-13}
     WORDS  &8.733&0.774&1.587&2.225&13.394&3.007&2.505&4.650&4.308&4.153&5.532&5.071\\
    
    \bottomrule
  \end{tabular}
  \caption{%
    Experimental results of our \WM{}-construction algorithms described in \S{\ref{sec:adaption_of_our_algorithms}}.
    The hardware and test instances are described in \S{\ref{sec:experiments}}.
    We measured the running time (in seconds) of the algorithms using one core ($t_1$) and 32 cores ($t_{32}$).
    The memory is given in bytes per byte of the input text when using one core ($\text{m}_1$) and 32 cores ($\text{m}_{32}$).
    Again, our algorithms based on domain decomposition use the corresponding sequential version of \emph{pcWM} or \emph{psWM} when run on one core. \label{tab:wm_experiments}}
\end{table}
In addition to the different order of the intervals, we also need to store the number of zeros.
To this end, we use the starting positions of the intervals
, as the number of zeros in any level is equal to even bit prefixes in the previous level.
This requires additional $\lceil\lg\sigma\rceil\lceil\lg n\rceil$ bits of space and \Oh{\sigma} time.

\begin{figure}[t]
    \centering
    \includegraphics[scale=.95]{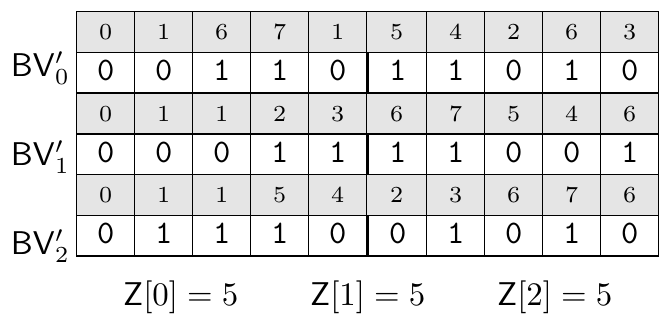}
    \caption{The \WM{} of our running example, $\Text{}=\mathtt{0167154263}$.
    The light gray (\protect\tikz[baseline=.25ex]{ \fill[black!20,draw] (0, 0) rectangle (.3, .3); }) arrays contain the characters represented at the corresponding position in the bit vector and are not a part of the \WM{}.
    The thick lines highlight the number of zeros at each level.\label{fig:example_WM}}
\end{figure}

\subsection{From the Wavelet Tree to the Wavelet Matrix.}
\label{sect:from_wt_to_wm}
We can also make use of these similarities by showing that \emph{every} algorithm that can compute a \WT{} can also compute a \WM{} in the same asymptotic time.

\begin{lemma}
  \label{lem:from_wt_to_wm}
  We can compute in-place an array \textsf{X} and a bit vector \textsf{U} with rank and select data structures in time \Oh{n+\sigma} and space $(n+\sigma)(1+o(1))+(\sigma+2)\lceil\lg n\rceil$ bits, such that $\BVT{\ell}\lbrack i\rbrack=\BV{}_\ell\lbrack j\rbrack$ with 
  $$j=%
  \begin{cases}%
    i&, \text{if~}\ell\leq1\\%
    \mathsf{X}\lbrack2^{\ell-1}-2+bp\rbrack+\text{off}&,\text{otherwise}%
  \end{cases}$$
  where $bp=\mathop{\mathrm{prefix}}(\ell,\rank_0(\mathsf{U},\select_1(\mathsf{U},i+1)))$ and $\text{off}=i-\rank_1(\mathsf{U},\select_0(\mathsf{U},bp\ll (\lceil\lg\sigma\rceil - \ell)))$, with $\ll k$ denoting a left bit shift (by $k$ bits), i.\,e., affixing $k$ zeros on the right hand side.
\end{lemma}

\begin{proof}
We require two auxiliary data structures for the transformation.
  The first one is the bit vector $\mathsf{U}$ of length $n+\sigma$ that stores the unary representation of the histogram of all characters in \Text{}.
  The second one is an array $\mathsf{X}$ of size $(\sigma + 2)\lceil\lg n\rceil$ bits, which at first is used for counting, and later on stores the starting positions of all intervals in the \WM{}.

  To compute $\mathsf{U}$ we first count the number of occurrences of all characters and store them in $\mathsf{X}$ such that $\mathsf{X}\lbrack i\rbrack=|\lbrace j\in[0,n)\colon \Text{j}=i\rbrace|$ for all $i\in[0,\sigma)$.
  Then, the unary histogram is given by $\mathsf{U}=1^{\mathsf{X}\lbrack \mathtt{0}\rbrack}\mathtt{01}^{\mathsf{X}\lbrack 1\rbrack}\mathtt{0}\dots\mathtt{1}^{\mathsf{X}\lbrack \sigma - 1\rbrack}$.
  In addition, we augment $\mathsf{U}$ with a rank/select data structure.
  All this requires \Oh{n+\sigma} time and $o(n+\sigma)$ bits space in addition to $\mathsf{U}$ and $\mathsf{X}$.

  Next, we want to compute the starting positions of the intervals in the \WM{} (i.\,e., fill the array $\mathsf{X}$ with its final content).
  We require those for intervals corresponding to bit prefixes of size $\ell$ with $\ell\in[1,\lceil\lg\sigma\rceil)$, i.\,e., for all but the first level of the \WM{}.
  To this end, we compute the number of occurrences of characters that share a bit prefix of size $\lceil\lg\sigma\rceil - 1$ in the first $\lceil\sigma/2\rceil-1$ positions of $\mathsf{X}$.
  With the histogram information still in $\mathsf{X}$, this can be done by setting $\mathsf{X}\lbrack i\rbrack=\mathsf{X}\lbrack 2i\rbrack+\mathsf{X}\lbrack 2i+1\rbrack$ for all $i\in[0,\lceil\sigma/2\rceil)$ in increasing order.
  We set all other positions of \textsf{X} to zero.
  Next, we compute the zero based prefix sum with respect to $\NBRP{\lceil\lg\sigma\rceil - 1}$ of the first $\lceil\sigma/2\rceil - 1$ entries of \textsf{X} and in the last $\lceil\sigma/2\rceil - 1$ entries of \textsf{X}.
  Here, ``respect to $\NBRP{\lceil\lg\sigma\rceil - 1}$'' means that character $\NBRP{\lceil\lg\sigma\rceil - 1}(i)$ follows character $\NBRP{\lceil\lg\sigma\rceil - 1}(i - 1)$ for all $i\in[1,\lceil\sigma/2\rceil)$.
  In the same fashion, we compute the starting positions of the intervals in all other levels.
  (By first computing the number of occurrences of bit prefixes of size $k$ using the ones of size $k+1$ and storing the zero based prefix sum with respect to $\NBRP{\lceil\lg k\rceil}$ in the rightmost free entries of \textsf{X}.)
  The $\sigma + 2$ entries (of size $\lceil\lg n\rceil$) in \textsf{X} are sufficient for this.
  Since the first entries of \textsf{X} can be empty (depending on $\sigma$), we finally move the starting positions to the left, such that the first starting position is stored in $\mathsf{X}\lbrack0\rbrack$.
  All this can be done in \Oh{\sigma} time without any additional space.
  Therefore, the construction of \textsf{U} (its augmenting rank/select data structure) and \textsf{X} requires \Oh{n+\sigma} time and $(n+\sigma)(1+o(1))+(\sigma+2)\lceil\lg n\rceil$ bits of space (including \textsf{U} and \textsf{X}).

  Now we need to answer queries asking for a position $j\in[0,n)$ in $\BV{}_\ell$ given a position $i\in[0,n)$ in $\BVT{\ell}$ for $\ell\in[0,\lceil\lg \sigma\rceil)$ in constant time, i.\,e., the position $j$ in the \WM{} corresponding to the position $i$ in the \WT{}.
  If $\ell\leq 1$ we know that $j=i$, because the bit vectors of the \WT{} and \WM{} are the same for the first two levels.
  Otherwise ($\ell>1$), the computation of the position $j$ consists of two steps.
  First, we determine the starting position of the interval in the \WM{} (using $\mathsf{X}$).
  Second, we compute the number of entries in the interval existing before $i$ (which is the same for \WM{} and \WT{}, as the intervals are the same):

  \begin{enumerate}
    \item We first need to identify the bit prefix of length $\ell$ corresponding to the interval containing $i$.
    Note that we are only interested in the bit prefix and not in the character $c$ corresponding to position $i$.
    There are at least $i-1$ (or none, if $i=0$) characters occurring in \Text{} whose bit prefix of length $\ell$ is at most \BitPre{\ell, c}.
    (There are more than $i-1$ characters if at least one character with bit prefix \BitPre{\ell, c} occurs after $c$ in \Text{}.)
    Therefore, $c^{\prime}=\rank_0(\mathsf{U}, \select_1(\mathsf{U}, i+1))$ has the same bit prefix of length $\ell$ as $c$, i.\,e., $bp=\BitPre{\ell, c^{\prime}}=\BitPre{\ell, c}$.
    Since we have stored all starting positions of the intervals on level $\ell$ in the \WM{} in $\mathsf{X}[2^\ell-2,2^{\ell+1})$ the starting position is $\mathsf{X}\lbrack 2^\ell-2+bp\rbrack$.
    \item Now we need to compute the offset of the position from the starting position of the interval.
    To do so, we compute the smallest character contained in the interval by padding the bit prefix with $\lceil\lg\sigma\rceil-\ell$ 0's giving us a value $r = \select_0(\mathsf{U}, bp\ll \lceil\lg\sigma\rceil-\ell)$.
    Next, we determine the number of 1's occurring before the $r$-th 0 in $\mathsf{U}$ to compute the offset, i.\,e., $\textit{off}=i-\rank_1(\mathsf{U}, r).$
  \end{enumerate}
  Since all operations used for querying require constant time and there is only a constant number of operations, the query can be answered in constant time.
\end{proof}

\section{Conclusions}
We presented new sequential and parallel wavelet tree (and matrix) construction algorithms.
Their unifying feature is their bottom-up approach, which saves repeated histogram computations per level from scratch and is also responsible for their space consciousness.
Our experiments showed that our new sequential algorithms are up to twice as fast as the previously known algorithms while requiring just a fraction of the memory (at most half as much).
In addition to the practical work, we also have shown how to (theoretically) adopt general \WT{}-construction algorithms to compute a \WM{} in the same asymptotic runtime.

The presented algorithms are the first practical \emph{parallel} \WM{}-construction algorithms.
It remains an open problem how to design parallel algorithms for wavelet \emph{matrices} that scale as well as the best one for wavelet \emph{trees} \cite{Labeit2016parallelWaveletTree}.

\section*{Acknowledgments}
We would like to thank Benedikt Oesing for implementing early prototypes of different sequential \WM{}-construction algorithms in his Bachelor's thesis \cite{Oesing2016} indicating promising approaches.
Further thanks go to Nodari Sitchinava (U.\ Hawaii) for interesting discussions on the work-time paradigm.


\bibliographystyle{plain}
\bibliography{lit}

\end{document}